\documentclass{article}
\usepackage{amsmath,amssymb,amsthm,amsfonts,amsthm}
\usepackage{natbib,graphicx}
\usepackage[margin=1in]{geometry}
\newtheorem{theorem}{Theorem}
\newtheorem{lemma}[theorem]{Lemma}

\newtheorem{corollary}[theorem]{Corollary}

\newtheorem{definition}{Definition}

\newcommand{\overbar}[1]{\mkern 1.5mu\overline{\mkern-1.5mu#1\mkern-1.5mu}\mkern 1.5mu}
\title{Parallel inversion of huge covariance matrices}
\author{Anjishnu Banerjee, Joshua Vogelstein, David Dunson}

\begin{document}
\maketitle
\begin{abstract}
An extremely common bottleneck encountered in statistical learning algorithms is inversion of huge covariance matrices, examples being in evaluating Gaussian likelihoods for a large number of data points. We propose general parallel algorithms for inverting positive definite matrices, which are nearly rank deficient.  Such matrix inversions are needed in Gaussian process computations, among other settings, and remain a bottleneck even with the increasing literature on low rank approximations.  We propose a general class of algorithms for parallelizing computations to dramatically speed up computation time by orders of magnitude exploiting multicore architectures.  We implement our algorithm on a cloud computing platform, providing pseudo and actual code. The algorithm can be easily implemented on any multicore parallel computing resource.  Some illustrations are provided to give a flavor for the gains and what becomes possible in freeing up this bottleneck.
\end{abstract}


\section{Introduction}
We consider a symmetric positive definite matrix $K$ of order $n$ where $n$ is very large; on the order of 10,000s to millions or more.  Our interest is in evaluation of $K^{-1}$, which cannot be computed sufficiently quickly using current algorithms.  Even if we could compute the inverse, substantial numeric instability and inaccuracies would be expected due to propagation of errors arising from finite machine precision arithmetic \citep{Trefethen97}.

Typically in statistical applications, one needs to evaluate expressions such as $K^{-1}A$, where $A$ is some matrix of appropriate order. Instead of directly computing $K^{-1}$, one popular approach is to consider the QR decomposition for $K=Q R$ and evaluate the expression, $R^{-1}Q^T A$ \citep{Press07}. The matrix $R$ is lower triangular and therefore $R^{-1}$ can be evaluated by backward substitution, which requires $O(n^2)$ flops instead of the $O(n^3)$ flops for usual inversion. QR decomposition is known to be relatively more stable and efficient than other standard algorithms \citep{Cox97QR}, a close competitor being the Cholesky decomposition. 
 
The problem is that QR decomposition for $K$ requires $O(n^3)$ computations, which is of the same order as that of inversion and therefore prohibitively expensive for large $n$. Even with QR decomposition being relatively stable, for very large $n$, finite precision numerical errors are large. This is accentuated when the matrix $K$ has a small spectrum, in the sense of fast decaying eigenvalues, as is typically obtained for covariance matrices obtained from smooth covariance kernels and in large least square problems among a host of other areas.  However, this small spectrum, while being the bane of full rank algorithms, makes low-rank algorithms highly accurate.

We propose a new class of algorithms combining ideas from recent developments in linear algebra to find approximate low-rank decompositions. These low-rank decompositions  provide several orders of magnitude improvement in speed while also improving accuracy. We propose a general blocking method to enable implementation in parallel on distributed computing architectures.  We also consider accuracy of these approximations and provide bounds which are obtained with high probability. Our approach amalgamates ideas from three recent but apparently unrelated developments in numerical linear algebra and machine learning, which we briefly outline below.

There has been an increasing literature on approximating a matrix $B$ of order $n \times n$ by $B \Omega$ where $\Omega$ is a matrix with random entries of order $n \times d$ with $d << n$ \citep{HM11}. Originally motivated by Johnson Lindenstrauss transforms \citep{Johnson84}, these results have been used in a host of application areas, including compressive sensing \citep{D06,CT06} and approximate Gaussian process regression \citep{Banerjee12}. Typically $\Omega$ is populated by random entries, such as draws from Gaussian or Rademacher distributions. The product $B \Omega$ involves $O(n^2d)$ flops and can be itself expensive. Recent developments have shown that structured random matrices, such as random Hadamard transforms, may improve the efficiency significantly, bringing down the number of flops required to the order of $n\log{d}$ \citep{Woolfe08}.  In our case, we show that with certain classes of structured random matrices, $\Omega$ of order $n \times d$, we have that $K\Omega$ has approximately the same column space as that of $K$ (in a sense to be made precise later), and that the product $K\Omega$ may be computed efficiently. In addition, such structured random matrices spread the information of the matrix $K$, so that we obtain bounds on the minimum and maximum eigenvalues of $K\Omega$, implying that approximate decompositions using $K\Omega$ are almost entirely devoid of usual inaccuracies arising from numerical conditioning.

Having efficiently obtained a tall skinny matrix $K\Omega$ from the square positive definite matrix $K$, we now consider recent developments which show that QR decompositions of such tall skinny matrices may be done efficiently in parallel \citep{Constantine11, Agullo10}. The key idea is to partition the tall matrix into blocks and efficiently combine $QR$ decomposition of each of these blocks. The main considerations are the choice of the number of blocks and column size $d$ of the projection matrix $\Omega$, which regulate computation time and accuracy. We provide some theoretical pointers while empirically demonstrating gains achieved on modest architectures. Approximate $QR$ decompositions are useful in obtaining other approximate matrix decompositions, such as approximate spectral decompositions or matrix products.  Several recent articles have focused on these issues and \citet{HM11} provides an excellent review for algorithms which switch between approximate decompositions. For illustration, we focus on the scenario when the principal interest is in evaluating a Gaussian type likelihood and fine tune the algorithms in this context. We finally consider performance versus other competitors.

\section{Notational Preliminaries}

We begin with some notations which we shall use through the rest of the article. In general we will represent matrices in the upper case, $A,B$ etc and row or column vectors in the lower case, $a,b$ etc. We shall use Frobenius and spectral matrix norms: $\|A\|_F=\sqrt{\sum_{i,j}a(i,j)^2}$ and $\|A\|_2=max\{\|Ax\|:\|x\|=1\}$ respectively, with the notation $\|A\|_\psi$ to imply the conclusion holds for both norms.  We begin with a  real-valued positive definite matrix $K \in \mathbb{R}^{n \times n}$ and $k(i,j)$ will denote the $(i,j)$th element of $K$. A spectral decomposition of $K=UDU^T$ can be partitioned as,
\[ K = \left[ U_m \, U_{(n-m)} \right] \left[ \begin{array} {cc}
                    D_{m} \,  & 0 \\
                    0 \,  & D_{(n-m)}
                   \end{array} \right] \left[ U_m \, U_{(n-m)} \right] ^{T}.
\]
Since $K$ is positive definite, the eigenvalues are positive and we assume without loss of generality that the diagonal matrix $D$ of eigenvalues contains them in descending order of magnitude and the eigenvector matrix $U$ is orthonormal. $D_{m}$ is therefore the matrix containing the $m$ largest eigenvalues and $U_m$ the corresponding eigenvector matrix. By the Eckart Young theorem \citep{S93}, the best rank $m$ approximation $K_m$ to $K$ in terms of minimizing $\|K-K_m\|_\psi$ is given by $K_m=U_m D_m U_m^T$.

 We shall focus on the case where the matrix $K$ has fast decaying eigenvalues. Some examples of problems where this happens is in large least squares optimization problems, $\|Ax-b\|$, with $K= A^T A$, or with $K$ being a covariance matrix generated for a Gaussian process at a dense set of locations with a smooth covariance function \citep{Banerjee12, Drineas11}. For some of our results, we will assume that eigenvalues of $K$ decay at an exponential rate, by which we assume that there exists positive constants $\lambda_1, \lambda_2$ such that $d_{i,i}\leq\lambda_1 e^{\lambda_2 i} \forall i$. $\lambda_1, \lambda_2$ shall be referred to as the proportionality and rate constants of exponential decay respectively.  

 An important consideration in this article will be the numerical stability of the algorithms and decompositions. One way in which the stability of a matrix decomposition can be measured is by how well the matrix in question is conditioned. The condition number of a positive semidefinite matrix $A$ is given by $c(A)=\sigma_l/\sigma_s$, where $\sigma_l,\sigma_s$ are the largest and smallest eigenvalues of the matrix respectively. For the positive definite matrix $K$ it would be $c(K)=d(1,1)/d(n,n)$.

\section{The Main Algorithm}

Our fast approximate inversion algorithm uses three different ideas from numerical techniques and randomized linear algebra. The first is to approximate the spectrum of the large matrix $K$ by post-multiplying it with a random matrix $\Omega$, the second is to consider incorporating special structure in the random matrix $\Omega$ for faster evaluation of the product and lastly, blocking strategies, to enable computation of decomposition of the product $K\Omega$ on multicore architectures. 

\subsection{Column space approximation}

 As mentioned, we shall be considering positive definite matrices $K$ with very fast decaying spectrums. In section \S$4$, we provide theoretical justification to show that eigenvalues of positive definite matrices, produced as discrete realizations of a wide class of positive definite kernels, decay very fast. For a positive definite matrix  $K$ with fast decaying spectrums, it is reasonable to expect that if $\Omega$ is a $n \times r$ matrix with $r \ll n$, populated with independent entries, then the product $K \Omega$ will capture most of the information in $K$, or approximate the column space of $K$. Such approximation schemes have been in focus in the rapidly expanding field of \emph{randomized linear algebra} (refer to \citet{HM11} for a review). \citet{Banerjee12} consider this method in the context of approximations for Gaussian processes for large data sets. In general, it is difficult to measure the amount of information captured in $K\Omega$ from $K$. One way to do this is to consider the generalized projection of $K$ onto $K\Omega$ and then consider differences.  Letting $Q$ be an orthonormal basis for the column space of $K\Omega$,  we may consider the error, $\|K-QQ^TK\|_{\psi}$, where $Q^TQ=I$, by virtue of orthonormality. In general, we may consider an approximation scheme with two objectives in mind.
\begin{enumerate}
 \item Fix a target rank $m$ and try to minimize the matrix norm error when using $\Omega$ having the target rank.
 \item Fix a target matrix norm error $\epsilon$ and try to achieve that error with high probability with the smallest possible $\Omega$.
\end{enumerate}
We shall consider each of objectives in our projection approximation designs.
 
The simplest possible choice for the random matrix $\Omega$ is a submatrix of columns from the $n \times n$ identity matrix (this submatrix is sometimes called a permutation submatrix, denoted by $P$). This amounts to choosing $r$ columns at random from the matrix $K$. Low rank approximations using a submatrix have been explored in a variety of contexts, including subset of regressors \citep{CR05} and least squares \citep{Drineas11}. With the advent of compressive sensing \citep{D06,CT06}, it has been shown that instead of restricting $\Omega$ to permutation submatrices, more general random matrices often have better performance \citep{Drineas11}. Common choices of the random matrix $\Omega$, which have been used in the compressive sensing and matrix approximation literature are:
\begin{enumerate}
 \item elements of $\Omega$ are independent and identically distributed Gaussian random variables appropriately scaled,
 \item elements of $\Omega$ are independent and identically distributed Rademacher variables,
 \item elements of $\Omega$ are independently sampled from the uniform distribution on the unit sphere.
\end{enumerate}
In general it has been shown that most choices, based on sampling from a centered distribution and then appropriately scaling, work in achieving accurate error bounds with high probability \citep{CT06}. 

\subsection{Approximations via structured random matrices}
The biggest drawback of the approximation techniques in \S$3.4.1$ is that the matrix product $K\Omega$ may itself be prohibitively expensive $O(n^2r)$  for very large $K$. Instead we may form the random matrix $\Omega$ such that the product $K\Omega$ can be evaluated faster, exploiting properties of the structured distribution.

\begin{definition}
\label{srmatrix}
We shall call a random matrix a structured random matrix if it is of the form $\Omega= cRTP$, where
\begin{itemize}
 \item $c$ is an appropriate scaling constant, such that the columns of $\Omega$ are orthonormal,
 \item $R$ is an $n \times n$ diagonal matrix whose diagonal elements are independent Rademacher entries,
 \item $T$ is an $n \times n$ appropriate orthogonal transform, facilitating the fast multiply,
 \item $P$ is an $n \times r$ permutation submatrix. 
\end{itemize}
\end{definition}

There are different choices of orthogonal transforms we may use, all of which facilitate the fast multiply, examples for $T$ being the discrete Fourier matrix, discrete cosine matrix, discrete Hartley matrix or the Walsh-Hadamard matrix. The Walsh-Hadamard matrix in particular has been in focus in recent literature \citep{Tropp11} and tight bounds on the approximation accuracy have been obtained. The Walsh Hadamard matrix of order $n \times n$ is defined as,
\[
 H_n = \left[
\begin{array}{cc}
H_{\frac{n}{2}} \, & H_{\frac{n}{2}}\\ 
H_{\frac{n}{2}} \, & -H_{\frac{n}{2}}
\end{array}
\right]
\quad \mbox{with}\quad
 H_2 = \left[
\begin{array}{cc}
1 \, & 1\\ 
1 \, & -1
\end{array}
\right].
\]

The Walsh-Hadamard and discrete Fourier transforms have some optimality with respect to preserving matrix coherence \citep{Tropp11,Gittens12}, but a disadvantage is Walsh-Hadamard transforms exist only for positive integers which are powers of $2$, while the discrete Fourier transforms may involve complex numbers. The alternatives, Hartley and discrete cosine transforms while being marginally suboptimal do not have these problems \citep{Avron10}. In the following section, we formalize the theoretical setting  and describe probabilistic error bounds for column space approximations using transforms with 
general orthonormal matrices (therefore valid for any of the aforementioned transforms). 

\subsection{Blocked decompositions for parallelization}
Our focus has been on obtaining a low rank rectangular matrix $\hat{K}$ of order $n \times r$, with $r \ll n$,  as an approximation to $K$, such that the column space of $\hat{K}$ is very close to the column space of $K$. To measure the approximation error, we used the error in column space approximation via projection, measured as, $||K-QQ^TK\|_{\psi}$, where $Q$ is an orthogonal basis for the column space of $\hat{K}$. Such a $Q$ is typically obtained from the $QR$ decomposition of $\hat{K}$, where $Q$ is an $n \times r$ matrix with orthonormal columns and $R$ is an upper triangular matrix (lower triangular in some conventions). Obtaining a $QR$ decomposition of an $n \times r$ matrix has computational cost $O(nr^2)$. With the matrix multiplication for forming $K\Omega$, the eventual computational cost is $O(n^2r)$ (for $n \gg r$), which may be too expensive for very large $n$. 

To further reduce computational burden, we may consider parallelizing the computations. Matrix multiplication is trivially parallelizable, while most matrix decompositions are not \citep{Choi94}. State-of-the-art linear algebra algorithms, for example in ScaLAPACK routines \citep{Blackford97}, $QR$, Cholesky or $SVD$ decompositions, involve dense manipulations of the whole matrix, which are not parallelizable or  cannot be blocked. In addition to computational cost, storing the full matrix $\hat{K}$ may be problematic in terms of memory requirements for very large $n$. However, in our case, we are interested in the $QR$ decomposition of a tall matrix, number of rows being much greater than the number of columns, for which blocking strategies have been recently developed, exploiting modern parallel computing  platforms like MapReduce. We elaborate on the blocking strategies next. 

For notational clarity assume that $n=8r$ and consider a partition of $\hat{K}$,
\begin{equation*}
\hat{K} = \left[
\begin{array}{c}
K_1 \\ \hline
K_2 \\ \hline
K_3 \\ \hline
K_4
\end{array}
\right]
\end{equation*}
In the above each $K_i$ is of size $2r\times r$. Denote the QR factorization of the small matrix $K_i=Q_i R_i$. Along the lines of the $TSQR$ factorization presented in \citet{Constantine11} we consider the following blocked factorization,
\begin{equation}
\label{block_qr}
\left[
\begin{array}{c}
K_1 \\ 
K_2 \\ 
K_3 \\
K_4 
\end{array}
\right]
=\left[
\begin{array}{cccc}
Q_1 & & & \\ 
 & Q_2 & &\\ 
 & & Q_3 &\\ 
&  &  & Q_4\\ 
\end{array}
\right]
\quad 
 \left[
\begin{array}{c}
R_1 \\ 
R_2 \\ 
R_3 \\
R_4 
\end{array}
\right] 
= \quad Q^b_1 R^b_1
\end{equation}
Each of the small matrices $Q_i$ being orthogonal, the left factor of the decomposition in \eqref{block_qr}, $Q^b_1$, is formed by stacking the small matrices $Q_i$ as diagonal blocks, in an orthogonal matrix. It is also an orthogonal basis for the column space of $\hat{K}$ and in many application areas it suffices to work with $Q^b_1$. In some other applications, we may require evaluation of the product $Q^b_1 x$, where $x$ is some vector of order $r \times 1$, which can be evaluated in a similar blocked fashion, utilizing the small $Q_i$'s. The matrix $R^b_1$, formed by stacking the small upper triangular matrices $R_i$, is not upper triangular and hence the decomposition $Q^b_1 R^b_1$ is not a $QR$ decomposition of $\hat{K}$. To achieve the $QR$ decomposition of $\hat{K}$ consider the following sequence of blocked decompositions for $R^b_1$,
\begin{equation}
\label{block_qr_it1}
\left[
\begin{array}{c}
R_1 \\ 
R_2 \\ \hline 
R_3 \\
R_4 
\end{array}
\right]
=\left[
\begin{array}{cc}
Q_5 &  \\ \hline
 & Q_6 \\ 
\end{array}
\right]
\quad 
 \left[
\begin{array}{c}
R_5 \\ \hline
R_6 
\end{array}
\right] 
= \quad Q^b_2 R^b_2,
\end{equation}
and finally the $QR$ decomposition of the $2r \times r$ matrix $R^b_2$,
\begin{equation}
 \label{block_qr_it2}
\left[
  \begin{array}{c}
R_5 \\ \hline
R_6 
\end{array}
\right] 
= \quad Q^b_3 R^b_3,
\end{equation}
where we can give the final $QR$ decomposition of $\hat{K}$ as, 
\begin{equation}
 \label{block_qr_final}
\hat{K} = (Q^b_1 Q^b_2 Q^b_3) R^b_3 = \hat{Q} R^b_3.
\end{equation}
$R^b_3$ by virtue of the construction is an upper triangular matrix, while $\hat{Q}$, a product of orthogonal matrices is itself orthogonal. A similar iteration can be carried out for any $n,r$ by appropriately adjusting the block sizes. In some instances, if only $(R^b_3)^{-1}$ is required, the matrix product for forming $\hat{Q}$ need not be evaluated at all.  On the other hand, if a product of the form $\hat{Q}^T x$ is required for some some vector $x$, the entire product maybe evaluated via blocking. We give more details about such strategies in our examples.

There are other ways of constructing a sequence of small blocked $QR$ decompositions leading to a $QR$ decomposition of the tall matrix $\hat{K}$. The scheme presented in \eqref{block_qr_final} is the one allowing for maximum assimilation of computations at the same time.  Depending on the architecture being used, this may not be the most optimal choice. Below we present another analogous blocked QR scheme, which allows for lesser number of computations to be carried on simultaneously, but possibly uses lesser communication between units of a multicore architecture. Assume, again for notational clarity that $n=8r$. The matrix $\hat{K}$ is now partitioned into $8$ blocks, $\{K_i\}_{i=1}^8$ of size $r \times r $ each. We operate on $\{K_i\}_{i=1}^4$ and $\{K_i\}_{i=5}^8$ simultaneously, via a sequence of small $QR$ decompositions in the following manner, 

\begin{equation}
\label{block_qr_alt}
\left[
\begin{array}{c}
K_1 \\ 
K_2 \\ 
K_3 \\
K_4 \\ \hline
K_5 \\
K_6 \\
K_7 \\
K_8
\end{array}
\right]
=
\begin{array}{c}
\left[
\begin{array}{cccc}
Q_1 & & & \\ 
 & I & &\\ 
 & & I &\\ 
&  &  & I 	
\end{array}
\right]
\quad 
 \left[
\begin{array}{ccc}
Q_2 & & \\ 
 & I  &\\ 
 & & I  
\end{array}
\right]
\left[
\begin{array}{cc}
Q_3 & \\ 
 & I 
\end{array}
\right]
\quad Q_4 R_4
\\ \hline
\left[
\begin{array}{cccc}
Q_5 & & & \\ 
 & I & &\\ 
 & & I &\\ 
&  &  & I 	
\end{array}
\right]
\quad 
 \left[
\begin{array}{ccc}
Q_6 & & \\ 
 & I  &\\ 
 & & I  
\end{array}
\right]
\left[
\begin{array}{cc}
Q_7 & \\ 
 & I 
\end{array}
\right]
\quad Q_8 R_8,
\end{array}
\end{equation}
where, $Q_1 R_1$ is the orthogonal matrix from the $QR$ decomposition of $K_1$, $Q_2$ is from the $QR$ decomposition of $\left[ \begin{array}{c} R_1 \\ K_2                                                                                                       
                                                                                                                     \end{array} \right]$, etc and $I$ denotes identity matrix of appropriate size. The matrices $R_1, R_2 \ldots$ are intentionally omitted from \eqref{block_qr_alt} for clarity. Analogously $Q_5,Q_6,\ldots$ are obtained from the partition $\{K_i\}_{i=5}^{8}$. To complete the decomposition, we make another $QR$ decomposition of $\left[ \begin{array}{c} R_4\\ R_8 \end{array} \right]$. This alternative blocking scheme involves minimum communication between units of the multicore architecture. The optimal schemes will usually be a mixture of the strategies \eqref{block_qr_it1} and \eqref{block_qr_alt}. For more variants of blocking algorithms to get $QR$ decompositions for tall matrices, we refer the readers to \citet{Constantine11,Agullo10}.

We now consider the theoretical gains possible from the blocked algorithm, while considering block size. Consider the ideal scenario, when there is no overhead or communication cost and each unit of the multicore architecture has equal speed. We consider the strategy corresponding to \eqref{block_qr_it1}. Each of the small $QR$ decompositions involve an $2r \times r$ matrix, which has $O(r^3)$ cost. At the first iteration, as in \eqref{block_qr_it1}, using the same number of units as the number of blocks, say $b$, we can perform all the computations in $O(r^3)$. In fact, it is easy to see that the computational order of $O(r^3)$ holds true for each of the subsequent iterations. The total number of iterations needed, following the first strategy would be $\lfloor \log_2 (b) \rfloor +1$, which brings the total computational cost to $O(\log_2 (b) r^3)$ for the $QR$ of $\hat{K}$. Forming the product $K\Omega$ and additional matrix multiplications (which will almost always be needed in applications after the $QR$ decomposition), bring the cost to $O(n^2r)$ without parallelization, whereas with the blocking algorithm, provided we perform the matrix multiplications with the same blocked structure, we achieve the whole computation in $O(\log_2(b) r^3)$. Following the first strategy and from the above discussion, in the ideal scenario, the number of blocks to be used is $b=\lfloor \frac{n}{2r} \rfloor +1$. The theoretical possible maximum speed-up, assuming no communication cost, is $O(\frac{b^2}{\log_2(b)})$, where we measure the speed-up from Amdahl's equation \citep{Mattson05}, as proportion of speed of the unparallelized computations to the speed of the parallelized computations. 

\section{Theory: Motivating results and approximation error bounds}

\subsection{Fast decay of spectrum of large positive definite matrices}
In this subsection we derive some motivating results justifying our assertion in \S$3.1$ that positive definite matrices obtained as dense discrete realizations of positive definite kernels will have a very fast decaying spectrum, enabling us to obtain good approximations using much lower rank matrices. The ideas are not entirely new, we adapt abstract results from theory of integral equations and stochastic series expansions to our context. To prepare the background, we start with the well-known Mercer's theorem for positive definite functions.

\begin{definition}
Let $D$ be a compact metric space and $C$ a function, $C:D\times D\to\mathbb{R}^+\cup\{0\}$. $C$ is said to be positive definite if for all $n$, scalars $c_1,\ldots,c_n \in \mathbb{R}$ and $x_1,\ldots,x_n \in D$, we have $\sum_i\sum_j c_i\,c_j\,C(x_i,x_j)>0$. It follows trivially from the definition that if $K$ is an $n \times n$ matrix with $k(i,j)=c(x_i,x_j)$ for $x_1,\ldots,x_n \in D$, then $K$ is a positive definite matrix.  
\end{definition}

With this definition, we give the following version of Mercer's theorem for the positive definite function $C(\cdot,\cdot)$ \citep{Kuhn87},
\begin{theorem} [Mercer's theorem]
 For every positive definite function $C(\cdot,\cdot)$ defined from $D \times D \to \mathbb{R}^+\cup\{0\}$, where $D$ is a compact metric space, there exists a real valued scalar sequence, $\{\lambda_i\}\in l_1$, $\lambda_1 \geq \lambda_2 \geq \ldots \geq 0$ and an orthonormal sequence of functions, $\{\phi_i(\cdot)\} \in L^2(D)$, $\phi_i(\cdot):D\to \mathbb{R}$, such that,
\begin{align*}
 C(x_1,x_2)=\sum_{i\in \mathbb{N}} \lambda_i \phi_i(x_1) \phi_i(x_2),
\end{align*}
$\forall x_1, x_2 \in D$ and this sequence converges in the mean square sense and uniformly in $D$. 
\end{theorem}

We omit the the proof of this theorem and refer the readers to \citep{Kuhn87}.Mercer's theorem is a generalization of the spectral decomposition for matrices extended to functions. $\{\lambda_i\}$, $\{\phi(\cdot)\}$  are referred to as the eigenvalues and eigenfunctions respectively of $C(\cdot,\cdot)$. Alternatively, suppose $D \subset \mathbb{R}^d$ and considering Lebesgue measure, we have the following integral equation for the eigenvalues and eigenfunctions, analogous to the definition of matrix eigenvalues and eigenfunctions,
\begin{align}
 \label{integral_eqn}
 \lambda_i \phi_i(y) = \int_D C(x, y)\phi_i(x)dx,
\end{align}
$\forall x \in D , i \in \mathbb{N}$.

The eigenvalues and eigenfunctions of $C(\cdot,\cdot)$ can be approximated using the eigenvalues and eigenvectors of the positive definite matrix $K$, obtained by evaluating $C(\cdot,\cdot)$ at any set of $n$ locations $x_1,\ldots x_n \in D$.  To see this, consider a discrete approximation of the integral equation (\ref{integral_eqn}), using the points $x_1, \ldots, x_n$, with equal weights in the quadrature rule,
\begin{align}
 \label{integral_approx}
\lambda_i \phi_i(y) \approx \frac{1}{n}\sum_{j=1}^{n} C(x_j, y)\phi_i(x_j)
\end{align}
Substituting $y=x_1,\ldots,x_n$ in (\ref{integral_approx}) we get a system of linear equations which is equivalent to the matrix eigenproblem, $K U^i=d(i,i) U^i, i =1,\ldots,n$, where $U^i$ is the $i^{\text{th}}$ row of the eigenvector matrix $U$. This approximation is related to the Galerkin method for integral equations and the Nystrom approximation \citep{Baker77, Delves88}. It also corresponds to finite truncation of the expansion from the Mercer's theorem. In general it can be shown that $d(i,i)/n$ and $\sqrt{n}\,u(i,j)$ converge to $\lambda_i$ and $\phi_i(x_j)$ respectively as $n \to \infty$ \citep{Baker77}.

The accuracy of the finite truncation of the Mercer expansion has been in recent focus and error bounds have been obtained depending on the degree of smoothness of the positive definite function \citep{Todor06, Schwab06}, in the context of stochastic uncertainty quantification. To borrow from these ideas and adapt the abstract results in our case, we begin with definitions quantifying the smoothness of the covariance functions. 

\begin{definition}
Let $C(\cdot,\cdot)$ be a positive definite function on $D \times D$, where $D$ is a compact metric space. We call $C(\cdot,\cdot)$ piecewise Sobolev $(p,q)$ smooth, for some $p,q \in \mathbb{N}$, if there exists a finite disjoint partition $\{D_j, j \in J\} $ of $D$, with $\overbar{D} \subset \cup_{j \in J} \overbar{D_j}$, such that for any pair $(j_1,j_2) \in J$, $C(\cdot,\cdot)$ is Sobolev $(p,q)$ on $D_{j_2} \times D_{j_2}$.
\end{definition}

\begin{definition}
Let $C(\cdot,\cdot)$ be a positive definite function on $D \times D$, where $D$ is a compact metric space. We call $C(\cdot,\cdot)$ piecewise analytic smooth, for some $p,q \in \mathbb{N}$, if there exists a finite disjoint partition $\{D_j, j \in J\} $ of $D$, with $\overbar{D} \subset \cup_{j \in J} \overbar{D_j}$, such that for any pair $(j_1,j_2) \in J$, $C(\cdot,\cdot)$ is analytic on $D_{j_2} \times D_{j_2}$.
\end{definition}

Commonly used covariance functions such as the squared exponential, $C(x,y)=c_1 exp(-c_2\|x-y\|^2)$ and the Matern function, depending on the value of the parameter $\nu$, fall into one of the above categories. In fact the squared exponential function is Sobolev smooth on any compact domain and we can give stronger results for decay of its eigenvalues. 

Using the couple of definitions categorizing smoothness, we give the following result, which is analogous to a functional version of the Eckart-Young theorem.
\begin{lemma}
\label{mercer_truncation}
Let $D \subset \mathbb{R}^d$. For any $m \in \mathbb{N}$, and let $\mathcal{F}_m$ be an $m$ dimensional closed subspace of positive definite functions on $D \times D$, where by $m$ dimensional we mean that $\exists$ sequence of orthonormal function $\{\psi_j(\cdot)\}_{j=1}^m$ on $D$ such that$\{\psi_j^2(\cdot)\}_{j=1}^m$ spans $\mathcal{F}_m$, with positive coefficients by virtue of positive definiteness. Let $C_{\mathcal{F}_m}(\cdot,\cdot)$ be the projection of $C(\cdot,\cdot)$ onto $\mathcal{F}_m$, and infimum of the errors in approximating $C(\cdot,\cdot)$ by $m$ dimensional functions $\in \mathcal{F}_m$, ie $E_m=\inf_{\mathcal{F}_m}\{\|C-C_{\mathcal{F}_m}\|^2, C_m \in \mathcal{F}_m \}$. Then,\\
(i) If $C(\cdot,\cdot)$ is piecewise Sobolev $(p,q)\,\, \forall (p,q) \,\, \exists $ for each $s>1$, a positive constant $c_s$ depending only on $C(\cdot,\cdot), s$, such that $E_m \, \leq \, c_s \sum_{j \geq m} j^{-s}$. Henceforth we call this bound $E_m^s$.\\ 
\newline
(ii) If $C(\cdot,\cdot)$ is piecewise analytic $\exists $ , positive constants $c_1, c_2$ depending only on $C(\cdot,\cdot)$, such that $E_m \leq c_s \sum_{j \geq m} c_1 \exp^{c_2 m^{1/d}}$. Henceforth we call this bound $E_m^a$.
\end{lemma}

\begin{proof} Let the Mercer expansion for $C(x,y)$ be $\sum_{j \geq 1} \lambda_j \phi_j(x) \phi_j(y) $. Define a random function, $f(x)=\sum_{j \geq 1} \eta_j \lambda_j \phi_j(x)$, where $\eta_j$ are independent standard normal random variables. Let $C_{\mathcal{F}_m}(x,y)=\sum_{j=1}^{m} \theta_j \psi_j(x) \psi_j(y)$ corresponding to the basis $\{\psi^2(\cdot)\}_{j=1}^{m}$ for $\mathcal{F}_m$. Similar to $f$, corresponding to $C_{\mathcal{F}_m}(\cdot,\cdot)$, define a random function, $f_{{\mathcal{F}_m}}(x)=\sum_{j \geq 1} \zeta_j \theta_j \psi_j(x)$, where $\zeta_j$'s are independent random normal variables. Defining $\|f-f_{\mathcal{F}_m}\| = E(f-f_{\mathcal{F}_m})^2$, and applying theorem $2.7$ in \citet{Schwab06}, we get that, $E_m=\sum_{j \geq m} \lambda_j$. Then the result follows by applications of Corollary $3.3$ and Proposition $3.5$ in \citet{Todor06}. \end{proof}

The above lemma quantifies the finite truncation accuracy of the Mercer theorem for smooth kernels. The bounds obtained are optimal and in general cannot be improved. The spaces ${\mathcal{F}_m}$'s can be made more general to encompass all square integrable bivariate functions, but the proof becomes more involved in that case and we omit it for the sake of brevity. In case of the squared exponential kernel, which is smooth of all orders, the finite truncation is even sharper and we give the following corollary for it:

\begin{corollary}
\label{squared_exp}
Let $C(x,y) = \theta_1 \exp^{-\theta_2\|x-y\|^2}$, for $x,y \in D$, compact $\subset \mathbb{R}^d$ and $\theta_1,\theta_2>0$. Using all notations as in lemma \ref{mercer_truncation}, $\exists$ positive constant $c_{\theta_1,\theta_2}$ such that $E_m \leq  c_{\theta_1,\theta_2} \sum_{j \geq m} \frac{\theta_2^{j^{1/d}}}{\Gamma(j^{1/d}/2)}$. Henceforth we shall call this bound $E_r^e$.
\end{corollary}

The proof is exactly similar to the proof of Lemma \ref{mercer_truncation} and an application of Proposition $3.6$ in \citet{Todor06}.

We now describe some results relating to the matrix eigenvalues, using the strength of Lemma \ref{mercer_truncation} above. 
%
\begin{theorem}
\label{decay_theorem}
 Let $K$ be the $n \times n$ positive definite matrix with $k(i,j)=C(x_i,x_j)$, with $x_i,x_j\in\{x_1,\ldots,x_n\} \subset D \subset \mathbb{R}^d$, compact. Also let $K=UDU^T$ be the spectral decomposition of $K$. Then,\\
(i) If $C(\cdot,\cdot)$ is piecewise Sobolev smooth $\forall (p,q)$, then $\exists$ for each $s>1$, a positive constant $c_s$ depending only on $C(\cdot,\cdot), s$, such that $d(m,m) \, \leq \, n\, c^{-s} m^{-s}$.\\
(ii) If $C(\cdot,\cdot)$ is piecewise analytic, then $\exists$ positive constants $c_1, c_2$ depending only on $C(\cdot,\cdot)$, such that $d(m,m) \leq n\, c_1 \exp^{c_2 m^{1/d}}$\\
(iii) In particular, for a squared exponential kernel, $C(x,y) = \theta_1 \exp^{-\theta_2\|x-y\|^2} \,\, \exists$ positive constant $c_{\theta_1,\theta_2}$ such that $d(m,m) \leq  n\, c_{\theta_1,\theta_2} \frac{\theta_2^{m^{1/d}}}{\Gamma(m^{1/d}/2)}$.
\end{theorem}

\begin{proof}
We begin with the discrete approximate solution of the integral equation, using the Galerkin technique, described in equation \eqref{integral_approx}. Note that this method coincides with the Raleigh-Ritz method with the identity matrix in \citet{Baker77}, since $C(\cdot,\cdot)$ is symmetric positive definite. Applying then theorem $3.31$ on page $322$ of \citet{Baker77}, we have $d(m,m) \leq n \lambda_m$, where $\lambda_m$ is the $m^{th}$ eigenvalue from the Mercer expansion of $C(\cdot,\cdot)$. The result then follows by a straight-forward application of Lemma \ref{mercer_truncation} above and Corollary $3.3$, Proposition $3.5$ in \citet{Todor06}. For the squared exponential kernel, it is straightforward application of Corollary \ref{squared_exp} above and Proposition $3.6$ in \citet{Todor06}.
\end{proof}

This theorem shows that for most covariance functions, the positive definite matrices that are generated by their finite realizations have eigenvalues that decay extremely fast, which was our assertion at the start of the section. In the simulated experiments we compute eigenvalues of some such covariance kernels and show that empirical results support our theory.
%
\subsection{Approximation accuracy and condition numbers}

We now present some results regarding the accuracy of our approximation algorithms. We shall be concerned with the column space approximation error, when we use $
\hat{K} = \Omega$ instead of $K$, measured as $\|K- Q Q^T K\|_{\psi}$, where $\Omega$ is a structured random matrix as in definition \ref{srmatrix} and $Q$ is orthonormal basis for the columns of $\hat{K}$. Such a $Q$ can be obtained from the $QR$ decomposition of $\hat{K}$.

\begin{theorem}
\label{accuracy_bound}
Let $K$ be the $n \times n$ positive definite matrix with $k(i,j)=C(x_i,x_j)$, with $x_i,x_j\in\{x_1,\ldots,x_n\} \subset D \subset \mathbb{R}^d$, compact. Let $\Omega$ be an $n \times r$ structured random matrix as formulated in definition \ref{srmatrix} and $Q$ be the left factor from the $QR$ decomposition of $K\Omega$. Choose $r,k$ such that $4[\sqrt{k}+\sqrt{8\ln{kn}}]^2\ln{K} \leq r \leq n$. With probability at least $(1-O(1/k))$, the following hold true, \\
(i) If $C(\cdot,\cdot)$ is piecewise Sobolev smooth $\forall (p,q)$, then (a) $\|K-QQ^T K\|_2 \leq (1+\sqrt{7n/r}) c_s r^{-s}$, where $c_s$ is a positive constant depending on $s$ for any $s>1$; (b) $\|K-QQ^T K\|_F \leq (1+\sqrt{7n/r}) E_r^s$. \\
(ii) If $C(\cdot,\cdot)$ is piecewise analytic, then (a) $\|K-QQ^T K\|_2 \leq n(1+\sqrt{7n/r}) c_1 \exp^{c_2 m^{1/d}}$, where $c_1,c_2$ are positive constants ; (b) $\|K-QQ^T K\|_F \leq n(1+\sqrt{7n/r}) E_r^a$. \\
(iii) In particular, for a squared exponential kernel, $C(x,y) = \theta_1 \exp^{-\theta_2\|x-y\|^2}$, then (a) $\|K-QQ^T K\|_2 \leq (1+\sqrt{7n/r}) c_{\theta_1,\theta_2} \frac{\theta_2^{m^{1/d}}}{\Gamma(m^{1/d}/2)}$, where $c_{\theta_1,\theta_2}$ is a positive constant depending on $\theta_1,\theta_2$; (b) $\|K-QQ^T K\|_F \leq (1+\sqrt{7n/r}) E_r^e$, \\
where $E_r^s, E_r^a, E_r^e$ are bounds as in Theorem \ref{mercer_truncation} and Corollary \ref{squared_exp}.
\end{theorem}

\begin{proof} First note that, 
\begin{align*}
P_{\hat{K}} = ((\hat{K})^T \hat{K})^{+} (\hat{K})^T,
\end{align*}
where $((\hat{K})^T \hat{K})^{+}$ denotes the Moore-Penrose inverse of $((\hat{K})^T \hat{K})$. Let $Q_k R_k$ denote the $QR$ decomposition of $\hat{K}$, so that we have,
\begin{align*}
P_{\hat{K}} K &= (R_k^T Q_k^T Q_k R_k)^{-1} Q_k R_k K\\
             &=  R_k^{-1} Q_k^T K.          
\end{align*}
Let $K_r$ be the best rank $r$ approximation to $K$ and let the spectral decomposition of $K=UDU^T$. Then from the Eckart-Young theorem, $\|K-K_m\|_2 = d(r,r)$ and 
$\|K-K_m\|_2 = \sum_{j=r}^n d(j,j)$. The result then follows by an application of our Theorem \ref{decay_theorem} and Theorem $11.2$ in \citet{HM11}.
\end{proof}

In addition to providing accurate column space approximations, these orthogonal transforms spread the information of the matrix $K$ and improve its conditioning, while preserving its geometry. By preserving the geometry, we mean preserving the norms of the eigenvectors - we explore more of this aspect empirically in the simulations. Any low rank approximation improves the conditioning, but it has been shown \citep{Gittens12,Avron10} that projections using the orthogonal transforms as above, improve conditioning substantially beyond what is achieved by just using any low rank approximation with high probability in special cases. \citet{HM11,Gittens12} obtain bounds on the eigenvalues of $A\Omega$ where $\Omega$ is a structured random matrix and $A$ is orthogonal. We are in interested in the stability of the numerical system, $\hat{K} R_k^{-1}$, as explained in the introduction of this article, where the $QR$ decomposition of $\hat{K}=Q_k R_k$. We present the following result, without proof, trivially modifying results from \citet{Gittens12, Avron10},

\begin{theorem}
Fix $0< \epsilon < (1/3)$ and choose $0<\delta<1, r$ such that $6 \epsilon^2 [\sqrt(r) + \sqrt{8 \ln{n/\delta}}]^2 \ln{2n/\delta} \leq r \leq n$. Then with probability at least $1 - 2\delta$, we have the condition number of the linear system of interest, $c(K R_k^{-1}) \leq \sqrt{\frac{1+ \epsilon}{ 1- \epsilon}}$ .
\end{theorem}

This improvement in conditioning number causes huge improvement in learning algorithms over other competing approaches as we demonstrate later.
%

\section{Illustrations}

\subsection{Simulation Examples}
We begin with empirical investigation into the eigenvalue decays of matrices realized from commonly used covariance kernels. We consider an equispaced grid of $100$ points in $[0,1]$, to generate the $100 \times 100$ covariance matrices with $100$ positive eigenvalues. We start with such a moderate size, $100$, to illustrate that even for this small size, eigenvalues decay extremely fast, as was hypothesized and theoretically proven earlier. 

This first one we use is the squared exponential kernel, 
\begin{align*}
C(x,y)=\theta_1 \exp^{-\theta_2\|x-y\|^2}
\end{align*}
The parameter $\theta_1$ is a scaling parameter and does not change the eigendirections. In our present simulation, we set $\theta_1=1$. The parameter $\theta_2$ controls the smoothness of the covariance, as it decays with the distance $\|x-y\|$. We consider a range of $\theta_2$ values, $0.05, 0.5, 1, 1.5,  2,  10$, respectively, considering negligible decay with distance and very fast decay with distance. The plot of the eigenvalues is shown in figure \ref{sqexp_decay}. The figure shows similar rates of decay across the range of values of the smoothness parameter $\theta_2$. Depending on the smoothness parameter, the value of the largest eigenvalue changes by several orders, but in all cases, by the $10th$ largest eigenvalue, the sequence has approximately converged to $0$, which indicates that for a low rank approximation, a choice of rank $10$ would be sufficient for our purposes. 

This second one we use is the Matern covariance kernel \citep{WR96}, 
\begin{align*}
C(x,y) = \theta_1\frac{1}{\Gamma(\nu)2^{\nu-1}}\Bigg(\sqrt{2\nu}\frac{\|x-y\|}{\theta_2}\Bigg)^\nu B_{\nu}\Bigg(\sqrt{2\nu}\frac{\|x-y\|}{\theta_2}\Bigg), 
\end{align*}
where $B_{\nu}$ is the modified Bessel function of the second kind. The parameter $\theta_1$ is a scaling parameter and does not change the eigendirections, analogous to $\theta_1$ of the squared exponential kernel. In our present simulation, we set $\theta_1=1, \theta_2=1$. The parameter $\nu$ controls the smoothness of the covariance, analogous to the parameter $\theta_2$ of the squared exponential kernel. As a special case, as $\nu \to \infty$, we obtain the squared exponential kernel.  In typical spatial applications, $\nu$ is considered to be within $[0.5,3]$, as the data are rarely informative about the smoothness levels beyond these. We consider a range of $\nu$ values, $0.5, 1, 1.5,  2, 2.5, 3$, respectively, considering negligible decay with distance and very fast decay with distance. The plot of the eigenvalues is shown in figure \ref{Matern_decay}. The figure shows similar rates of decay across the range of values of the smoothness parameter $\nu$, exactly as exhibited with the squared exponential kernel. Once again, a low rank approximation of rank $10$ would suffice in this case for approximating the $100 \times 100$ matrix.

Using the same set of simulations, we compute conditioning numbers, in turn for each of the covariance kernels. We compute the conditioning number of the the full $100 \times 100$ covariance matrix and its best rank $m$ approximations for $m=5,10,15,20,50$ respectively. Results for the squared exponential kernel are tabulated in table \ref{sqexp_cond} and for the Matern covariance kernel in table \ref{Matern_cond}. The full covariance matrix is extremely ill-conditioned in each case, the best rank $m$ approximations improve the conditioning, as is to be expected. Even with the best rank $m$ approximations, the conditioning numbers are still very large indicating very unstable linear systems. In the table we omit the exact digits for the very large numbers and just indicate their orders in terms of powers of $10$. In case of the squared exponential kernel, condition numbers decrease from left to right and from top to bottom. This pattern is in general not true for the condition numbers of the Matern kernel, whose condition numbers are several orders smaller than that for the squared exponential but still to large for it to be stable linear systems. In general for the Matern kernel, as $\nu$ becomes larger, the kernel becomes smoother, the condition numbers are larger and the eigenvalues decay faster.

We next move to the simulations to consider the effect of blocking and gain in efficiency from parallelization. For large sample sizes, we will get very poor and time consuming inverse estimates, as demonstrated by the large conditioning numbers from the simulations of the previous section. To circumvent this, we start therefore with a known spectral decomposition and apply our approximations, pretending that the decomposition was unknown. We consider the known spectral decomposition as $K=EDE^T$, where $E$ is an orthonormalized matrix with randomly generated iid Gaussian entries. $D$ is the diagonal matrix of eigenvalues in decreasing order or magnitude and for this simulation example we assume exponential decay with scale $=1$ and rate $0.01$. We generate this covariance matrix $K$ for sample sizes $n=1000,5000,10000,50000$ respectively and apply our blocked approximate inversion algorithm. For the random matrix $\Omega$ used for the projection, we use three different choices, a matrix with scaled iid standard Gaussian entries, a structured random Hartley transform and a structured random discrete cosine transform. We have a total $64$ cores at our disposal and to get a flavor of the gains possible by parallelization, we use $8,16,32,64$ cores in turn and measure the gain in efficiency as the ratio of time taken to the time for the algorithm without parallelization.

We report the gain in efficiency in figure \ref{eff_gain}. The figure reveals the efficiency gain by using the blocked algorithm. It shows that efficiency gain increases with number of cores in the larger sample sizes. Specifically when the sample size is $50,000$, it seems that more than $64$ cores could have further increased the efficiency. The gains reported are obviously lower than the theoretical maximum possible gains, but are still substantial and have potential of further increase with larger number of cores in huge sample sizes. 

\subsection{Real data examples}
We consider a real data examples with a very large data set, to demonstrate the large data handling power and efficiency gain via our approach. We consider subsets of the actual data available and variables such that the data fit into a Gaussian process framework.

The example we consider is a subset of the Sloan sky digital survey \citep{Sloan07}. In the survey, redshift and photometric measurements were made for a very large number of galaxies. The photometric measurements consist of $5$ band measurements and are available for most of the galaxies, while the redshift measurements are not available for a large number of galaxies due to spectroscopy constraints. The main interest is therefore in predicting the red-shifts, given the other 5 measurements. For this example we consider a subset comprising of $180,000$ galaxies, whose photometric measurements and red-shift measurements are known. We hold out $30,000$ galaxies and pretend that redshifts are unknown for these galaxies, while using the remaining $150,000$ to train our Gaussian process model. 

For the computations we use both a squared exponential kernel and a Matern kernel. The data are scaled and centered before calculations so that we are not concerned with inference regarding the $\theta_1$ parameters for the covariance kernel. For the Matern kernel, we use a discrete uniform prior for the parameter $\nu$ based on $100$ equispaced points on the grid $[0.5,3]$. For the squared exponential kernel, we use a $100$ trials, each having a  random selection $1,000$ points from the training set to set a band for $\theta_2$, to estimate range of covariance of the data. It appears that it is sufficient to consider $\theta_2$ in the range $[0.05,1]$. We therefore use a grid on $100$ uniformly spaced points in $[0.05,1]$ and analogous to the prior for $\nu$, place a discrete uniform prior on $\theta_2$ of the squared exponential kernel.

For data sets of this size it becomes extremely inefficient to use predictive processes or knot based approaches per se, without blocking or parallelization and knot selection is almost practically infeasible to attempt. We work with the modified predictive processes approach of \citet{FB09}, without knot selection as competitors to our approach. The modified predictive processes is conceptually equivalent to the fully independent training conditional approach \citep{CR05} or the Gaussian process approximation with pseudo inputs, with \citep{SG06}, with the same priors for the parameters as our approach, so that the only difference is in the way the covariance inversion is done. 

For each of the approaches, including ours, one issue is how to select rank of the approximate projections. We use an approximate estimate of the covariance decay to come up with the low rank projection to be used. To calculate this we select $1,000$ data points spread out across the training set, in the following approximate manner, first select the pair of points which are most distant from each other, then select the point which is almost equidistant from the pair selected, and so on, till we have a well spread out selection. Here distance is measured as the Euclidean distance between the $5$ variate photometric measurement vectors. Using these $1,000$ data points, we calculate the prior covariance matrix, using the Gaussian covariance kernel and the Matern kernel for each value of the smoothness parameter on its grid of values (Our simulation examples have revealed that the decay of eigenvalues depends on this smoothing parameter). We then estimate the rank such that Frobenius norm error in using the best low rank approximation from the Eckart Young theorem for the full covariance matrix would be no more than $0.001$ (bearing in mind that this is an estimate of the best possible error, actual error in different projections will be more). In the computations for our approach, we use a different random projection to the targeted rank for each grid value of the smoothing parameter, while for the knot based approaches, a different random selection of knots is used. 

For our approach, we use $3$ different choices of the projection matrix $\Omega$, a matrix with scaled iid standard Gaussian entries, a structured random matrix with the Hartley transform and a structure random matrix with the discrete cosine transform. We use blocking as in our main algorithm for each of the three types of projection matrices on all $64$ cores of the distributed computing structure. The knot based approach has no obvious method of blocking or for it to be used on a grid computing environment. The MCMC algorithm in each case is run for $10,000$ iterations, with the first $2,000$ discarded as burn-in. Usual diagnostics do not show evidence of non-convergence of the MCMC algorithms and we report the effective sample size for the smoothing parameter as obtained from the CODA package in R \cite{PC06}. Time taken by each method is computed and reported.  The time calculations are made by recording the exact CPU times for iterations, excluding the time for Frobenius norm accuracy and conditioning number calculations and also excluding the time taken to set up parallel Matlab workers via matlabpool. We measure predictive accuracy, of the predictions in the hold out set, based on the relative mean squared error.

We tabulate the results from the experiment in the table \ref{realdata}.  The performance of any of the projection based approaches is substantially better than the knot based PP in terms of any of the parameters reported, predictive accuracy, time taken or the effective sample size of the smoothing parameter. Time taken by the PP in case of this large dataset was several days, as compared to a few hours for the blocked projection approaches. The improvement in predictive accuracy can probably be attributed to much better conditioned stable linear systems being used inherently and hence improved inference. The discrete cosine transform and Hartley transform perform marginally better than the random Gaussian projection, however the structured random transforms show marked improvement in the time taken, possibly due to the faster matrix multiplies. There is little to choose between the discrete cosine transform or the discrete Hartley transform.

\section{Discussion}
In this article we present a new method of blocked approximate inversion for large positive definite matrices, using subsampled random orthogonal transforms, motivated by several recent developments in random linear algebra. The blocking strategy comes from development of parallel $QR$ algorithms for tall matrices and in our case tall matrices are obtained from large positive definite matrices as the first projection step of a low rank approximation. We have also presented a comprehensive set of theoretical results quantifying the decay of eigenvalues of covariance matrices, which in turns helps bound the errors from the low rank projections. The examples show marked improvement in terms of numerical stability, prediction accuracy and efficiency and should be routinely applicable in a wide variety of scenarios.

One interesting future direction is the investigation of several structured random projections taken together. One may break up a target rank into several parts, the first of which is spent on a random projection, and the remaining pieces are learnt through the direction of maximum accuracy. Both theoretical (finding projections in the directions of maximum accuracy, which may be the direction of maximum information descent) and empirical investigation of this issue is currently being investigated.

\section*{Tables and figures}
\begin{table}[ht]
\caption{Table of condition numbers for the squared exponential kernel for different values of the smoothing parameter $\theta_2$ versus truncation levels, for truncating to the best possible approximation according to the Eckart-Young theorem. The rows represent the levels of truncation and the columns, the values of the smoothing parameter.}
\label{sqexp_cond}
\begin{tabular}{|c|c|c|c|c|c|c|}
& $\theta_1=0.05$ & $\theta_2=0.5$ & $\theta_2=1$ & $\theta_2=1.5$ & $\theta_2=2$ & $\theta_2=10$\\
\hline
full, $m=100$ & $ O(10^{20})$ & $O(10^{19})$ & $O(10^{19})$ & $O(10^{19})$ & $O(10^{18})$ & $O(10^{18})$\\
$m=50$ & $O(10^{17})$ & $O(10^{17})$ & $O(10^{17})$ & $O(10^{17})$ & $O(10^{17})$ & $O(10^{16})$\\
$m=20$ & $O(10^{17})$ & $O(10^{17})$ & $O(10^{16})$ & $O(10^{16})$ & $O(10^{16})$ & $O(10^{15})$\\ 
$m=15$ & $O(10^{17})$ & $O(10^{16})$ & $O(10^{16})$ & $O(10^{16})$ & $O(10^{16})$ & $O(10^{10})$\\ 
$m=10$ & $O(10^{16})$ & $O(10^{15})$ & $O(10^{13})$ & $O(10^{11})$ & $O(10^{10})$ & $O(10^{5})$\\
$m=5$  & $O(10^{9})$ & $O(10^{6})$ & $O(10^{4})$ & $O(10^{4})$ & $O(10^{3})$ & $29.89$\\
\hline  
\end{tabular}
\end{table} 

\begin{table}[ht]
\caption{Table of condition numbers for the Matern kernel for different values of the smoothing parameter $\nu$ versus truncation levels, for truncating to the best possible approximation according to the Eckart-Young theorem. The rows represent the levels of truncation and the columns, the values of the smoothing parameter.}
\label{Matern_cond}
\begin{tabular}{|c|c|c|c|c|c|c|}
& $\nu=0.5$ & $\nu=1$ & $\nu=1.5$ & $\nu=2$ & $\nu=2.5$ & $\nu=3$\\
\hline
full, $m=100$ & & & & &\\
full, $m=100$ & $ O(10^{3})$ & $O(10^{5})$ & $O(10^{7})$ & $O(10^{9})$ & $O(10^{11})$ & $O(10^{13})$\\
$m=50$ & $O(10^{3})$ & $O(10^{5})$ & $O(10^{6})$ & $O(10^{8})$ & $O(10^{9})$ & $O(10^{11})$\\
$m=20$ & $O(10^{2})$ & $O(10^{4})$ & $O(10^{5})$ & $O(10^{6})$ & $O(10^{7})$ & $O(10^{9})$\\ 
$m=15$ & $O(10^{2})$ & $O(10^{3})$ & $O(10^{4})$ & $O(10^{5})$ & $O(10^{6})$ & $O(10^{7})$\\ 
$m=10$ & $O(10^{2})$ & $O(10^{3})$ & $O(10^{3})$ & $O(10^{4})$ & $O(10^{5})$ & $O(10^{5})$\\
$m=5$  & $38.18$ & $O(10^{2})$ & $O(10^{2})$ & $O(10^{2})$ & $O(10^{3})$ & $O(10^{3})$\\\hline  
\end{tabular}
\end{table}

\begin{table}[ht]
\caption{Results from the real data experiment. Columns are the type of experiment, PP corresponds to the modified predictive process approach, RP corresponds to the projection method with a Gaussian projection matrix, HP corresponds to a structured random projection with the Hartley transform, HC corresponds to a structured random projection with the discrete cosine transform. Time taken is measured as relative time, taking the time taken by HC to be 1. RMSE is relative mean squared error, ESS stands for effective sample size.}
\label{realdata}
\begin{tabular}{|c|c|c|c|c|}
& PP & RP & HP & HC\\
\hline
RMSE, Sq Exp & 70.63 & 19.87 & 14.64 & 15.72 \\
Relative time, Sq Exp & $O(10^3)$ & 10.79 & 1.57 & 1 \\
Avg Condition No, Sq Exp & $O(10^7)$ & $O(10^3)$ & $O(10^2)$ & $O(10^2)$  \\
Avg Frobenius norm error, Sq Exp & 0.55 & 0.04 & 0.07 & 0.05\\ 
ESS, Sq Exp & 237 &  833 & 1991 & 1875  \\
RMSE, Matern & 35.61 & 21.27 & 20.83 & 20.97 \\
Relative time, Matern & $O(10^4)$ & 32.30 & 0.91 & 1\\
Avg Condition No, Matern & $O(10^4)$ & $O(10^3)$ & $O(10^2)$ & $O(10^2)$ \\
Avg Frobenius norm error, Matern & 1.37 & 0.62 & 0.18 & 0.39 \\ 
ESS, Matern & 569 & 1322 & 1219 & 1794 \\
\hline
\end{tabular}
\end{table}

\begin{figure}[ht]
\begin{flushleft}
\includegraphics[width=\textwidth, height=0.8\textheight]{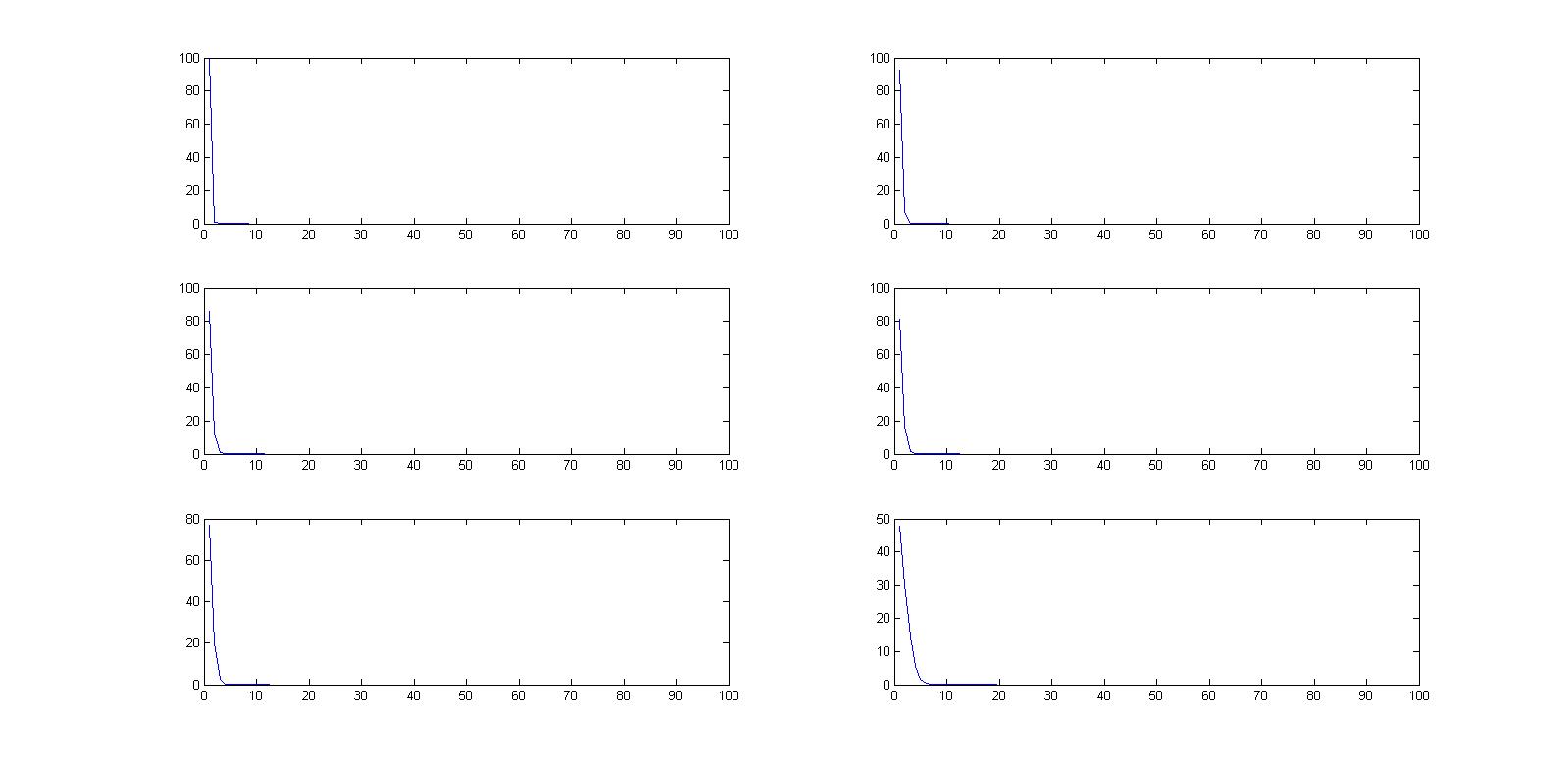}
\end{flushleft}
\caption{Decay of Eigenvalues: Panel representing decay of eigenvalues in the squared exponential covariance kernel. We plot the $100$ eigenvlaues in decreasing order of magnitude, the x-axes represent the indices, the y-axes the eigenvalues. Top left panel, top right, middle left, middle right, bottom left, bottom right are for values of the smoothness parameter $\theta_2 = \, 0.05, \, 0.5, \, 1, \, 1.5, \, 2, \& 10$ respectively.}
\label{sqexp_decay}
\end{figure}

\begin{figure}[ht]
\begin{flushleft}
\includegraphics[width=\textwidth, height=0.8\textheight]{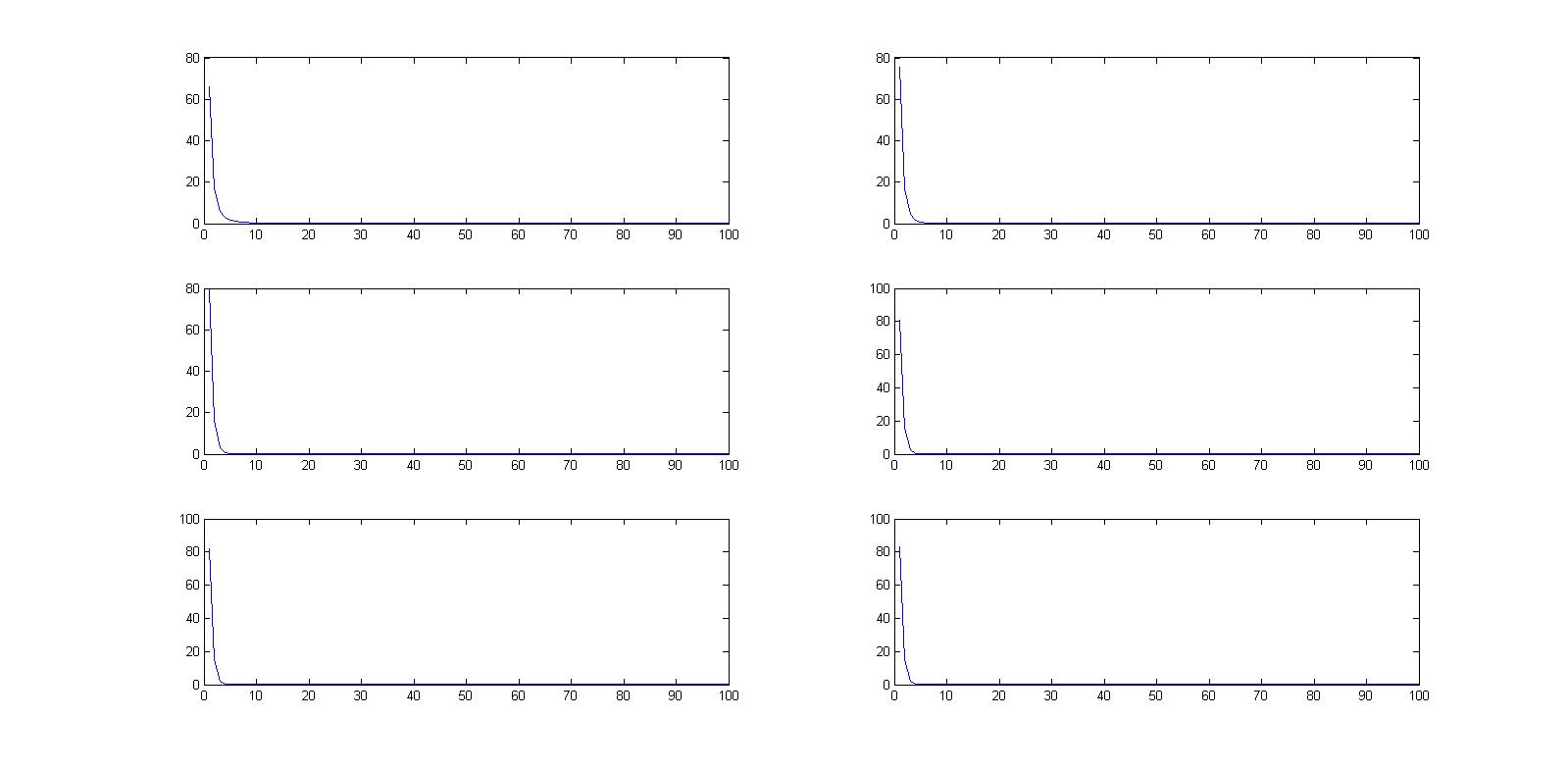}
\end{flushleft}
\caption{Decay of Eigenvalues: Panel representing decay of eigenvalues in the Matern covariance kernel. We plot the $100$ eigenvlaues in decreasing order of magnitude, the x-axes represent the indices, the y-axes the eigenvalues. Top left panel, top right, middle left, middle right, bottom left, bottom right are for values of the smoothness parameter $\nu = \, 0.5, \, 1, \, 1.5, \, 2, \, 2.5, \, \& 3$ respectively.}
\label{Matern_decay}
\end{figure}

\begin{figure}[ht]
\begin{flushleft}
\includegraphics[width=\textwidth, height=0.8\textheight]{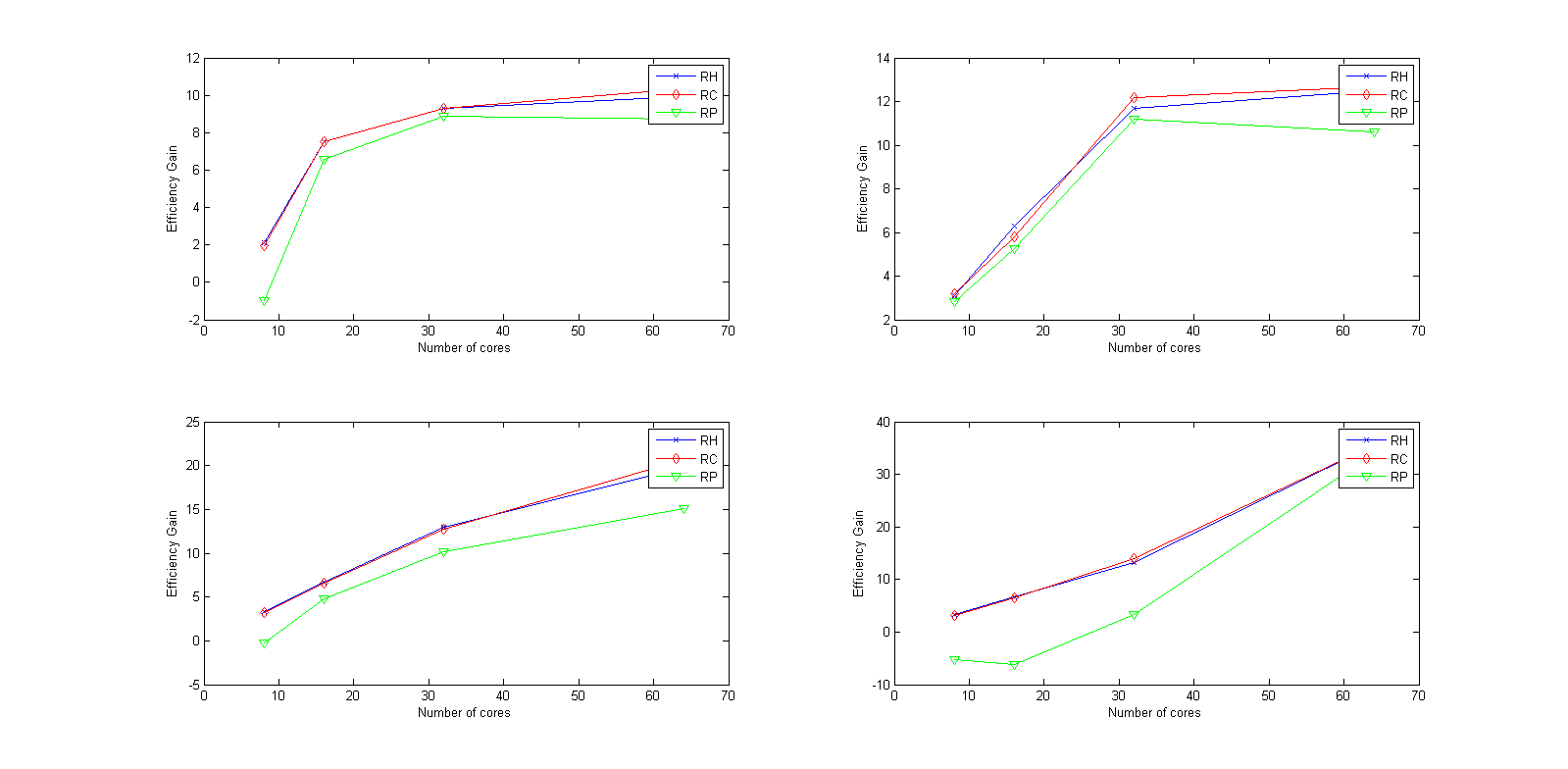}
\end{flushleft}
\caption{Gain in efficiency by using increasing number of cores in a parallel computing environment: Top left panel, top right, bottom left, bottom right are for sample sizes $n = \, 1000, \, 5000, \, 10000, \, 50000$ respectively. Superimposed on each panel are the gains by using a Hartley transform (RH), discrete cosine transform (RC) and a scaled random Gaussian projection (RP).}
\label{eff_gain}
\end{figure}

\bibliographystyle{jasa}
\bibliography{mybib}

\end{document}